\documentclass[11pt]{article}
\usepackage{amsmath,amsthm,amssymb}
\usepackage[numbers]{natbib}
\usepackage[margin=1in]{geometry}
\usepackage{tikz}
\usetikzlibrary{arrows.meta}
\usepackage{enumitem}
\usepackage[ruled,vlined]{algorithm2e}

\newtheorem{thm}{Theorem}
\newtheorem{lem}[thm]{Lemma}
\newtheorem{clm}[thm]{Claim}

\newtheorem{conj}[thm]{Conjecture}
\numberwithin{thm}{section}
\newtheorem{observation}{Observation}
\newtheorem*{remark*}{Remark}
\newtheorem*{notation*}{Notation}
\newtheorem*{observation*}{Observation}
\newtheorem*{theorem*}{Theorem}
\newtheorem*{def*}{Definition}
\newcommand{\mbs}{\boldsymbol}

\makeatletter
\renewcommand\@makefntext[1]{%
  \noindent\makebox[0em][r]{\@makefnmark}#1}
\makeatother

\begin{document}
\bibliographystyle{plainnat}

\begin{center}
\vspace{2in}
{\Large\bf One Dollar Each Eliminates Envy}\\ [.6cm]
\end{center}

\renewcommand*{\thefootnote}{\fnsymbol{footnote}}

\begin{center}
{\sc J. Brustle},
{\sc J. Dippel},
{\sc V.V. Narayan},
{\sc M. Suzuki}
\ and\ 
{\sc A. Vetta}
\footnote{\tt \{johannes.brustle,jack.dippel,vishnu.narayan,mashbat.suzuki\}@mail.mcgill.ca, adrian.vetta@mcgill.ca} \\[0.3cm]
McGill University
\\[0.3cm]
\today
\end{center}

\renewcommand*{\thefootnote}{\arabic{footnote}}
\addtocounter{footnote}{-1}


\begin{abstract}
    We study the fair division of a collection of $m$ indivisible goods amongst a set of $n$ agents. 
    Whilst envy-free allocations typically do not exist in the indivisible goods setting, 
    envy-freeness can be achieved if some amount of a divisible good ({\em money}) is introduced. Specifically, \citet{HS19}
  showed that, given additive valuation functions where the marginal value of each item is 
    at most one dollar for each agent, there always exists an envy-free allocation requiring a subsidy of at most $(n-1)\cdot m$ dollars. The 
    authors also conjectured that a subsidy of $n-1$ dollars is sufficient for additive valuations. We prove this conjecture. In fact, a subsidy of
    at most one dollar per agent is sufficient to guarantee the existence of an envy-free allocation. 
   Further, we prove that for general monotonic valuation functions an envy-free allocation always exists with a subsidy of at most $2(n-1)$ dollars
   per agent. In particular, the total subsidy required for monotonic valuations is independent of the number of items.
\end{abstract}

\section{Introduction} \label{s:introduction}
We consider the fair division of $m$ indivisible items amongst $n$ agents.
Specifically, we desire an allocation where no agent is envious of any other; that is, the value each agent has for its own allocated bundle is
at least as great as its value for the bundle of any other agent.
This concept, called {\em envy-freeness}, was introduced by \citet{Fol67}.
For divisible goods, \citet{Var74} explained how to obtain envy-free allocations using the theory of general equilibria: 
simply share each good equally amongst the agents and then find a competitive equilibrium. Thus envy-free allocations exist for classes of
valuation functions where competitive equilibria are guaranteed to exist.

Unfortunately for indivisible goods it is easy to see that envy-free allocations do not exist in general. 
For example, if the number of agents exceeds the number of items, then in every allocation 
there is an agent who receives an empty bundle.
In classical work, \citet{Mas87} asked if this impossibility result could be circumvented by the addition of a single 
divisible good, namely {\em money}. If so, how much money is needed to eradicate all envy? 
He considered the case of a market with $n$ agents and $m=n$ goods where each agent can be allocated {\em at most} one good, 
and has, without loss of generality, a value of at most one dollar for any specific good.
\citet{Mas87} then showed that an envy-free allocation exists with the addition of $n-1$ dollars into the market. 
But what happens in the general setting where the number of agents and number of goods may differ and where
agents may be allocated more than one good? The purpose of this paper is to understand this case of {\em multi-unit demand} valuations.
In this setting, \citet{HS19} proved that $m\cdot (n-1)$ dollars suffice to support an envy-free allocation when 
the agents have {\em additive} valuation functions. Further, they conjectured that, as with the unit-demand setting, 
there always exists an allocation for which $n-1$ dollars suffice.

The main result in this paper is the verification of this conjecture: for additive valuation functions, 
precisely $n-1$ dollars is sufficient to guarantee the existence 
of an envy-free allocation. In fact, our result is stronger in several ways. First,  
not only is the subsidy at most $n-1$ dollars in total but each agent receives {\em at most} one dollar in subsidy.
Secondly, this allocation is also \textit{envy-free up to one good (EF1)} -- this settles a second conjecture from \cite{HS19}. 
Thirdly, the allocation is {\em balanced}, that is, the cardinalities of the allocated bundles differ by at most one good.
Furthermore, this envy-free allocation can be constructed in polynomial time.

We also study the case of general valuation functions.
Requiring only the very mild assumption that the valuation functions are monotone,
we prove the perhaps surprising result that envy-free solutions still exist with a subsidy amount that
is {\em independent} of the number of goods $m$. Specifically, we prove that there is
an envy-free allocation where each agent receives a subsidy of at most $2(n-1)$ dollars, which is a 
total subsidy of $O(n^2)$.
Here the envy-free allocation can be constructed in polynomial time given a valuation oracle.

\subsection{Related Work}
\textit{Fair division} has been extensively studied over the past six decades. 
The concept of a fair allocation was formally introduced by \citet{Ste48} via the {\em cake-cutting problem}:
how can a heterogeneous cake be fairly divided among a set of agents? To address this question, it is necessary 
to first define fairness. The fairness objective of \citet{Ste48} was \textit{proportionality}. An allocation is proportional if every agent is allocated a 
bundle (or piece of cake) of value at least $\frac{1}{n}$ of its total value for the grand bundle (entire cake). 
For cake-cutting, and divisible goods in general, when the valuations are additive, envy-freeness implies 
proportionality. This is because, for any agent and 
any partition of the cake into $n$ pieces, some piece must be worth at least $\frac{1}{n}$ of the whole cake to that agent.
Thus envy-freeness is a stronger fairness guarantee than proportionality. 
Another classical fairness measure is {\em equitability}, where all agents should receive bundles
of the same value. In the case of divisible goods, \citet{Alon87} showed for additive continuous valuation functions that
allocations exist that satisfy proportionality, equitability and envy-freeness simultaneously.
Algorithmic methods to obtain envy-free cake divisions for any number of agents are also known; see, 
for example, \citet{BT95}.

More recently, \citet{Bud11} introduced the \textit{maximin share guarantee}
inspired by the cut-and-choose protocol. Assume an agent partitions the items into $n$ bundles and then receives  
the lowest-value bundle. The corresponding value that the agent obtains by selecting its optimal partition is called 
its {\em maximin} share. The fairness objective then is to find an allocation where every agent receives a bundle
of value at least its maximin share.

Unfortunately, for indivisible goods, there are examples where 
proportionality, envy-freeness, equitability and the maximin share guarantee
are all impossible to achieve. Consequently, there has been much focus on 
approximate fairness guarantees. One natural approach is the design of approximation algorithms 
for the maximin share problem. An alternative guarantee is via EF-$k$ allocations \cite{Bud11}, where
an agent has no envy provided $k$ goods are removed from the bundles of the other agents. Of special interest 
are the \textit{envy bounded by a single good}, or EF1 
allocations. \citet{LMM04} showed that, when the valuation functions are 
monotone, an EF1 allocation exists and it can be computed in polynomial time. A large body of recent work on 
the fair allocation of indivisible goods has focused on achieving these types of approximation
guarantee \cite{GHS18,KPW18,BCF19,CKM19}.

A parallel line of research considers the use of money in the fair allocation of indivisible goods. This is 
motivated by the \textit{rent division} problem, where the goal is to allocate $n$ indivisible goods among $n$ 
agents and divide a fixed total cost, $i.e.$ the {\em rent}, amongst the agents. \citet{Su99} showed that, under mild 
assumptions, \textit{rental harmony} can be achieved: there is an envy-free division of the goods and the rent. 
The majority of the literature in this area considers the setting with $n$ unit-demand agents, $m=n$ indivisible items 
and one divisible good, akin to {\em money}. \citet{Sve83}~showed that an envy-free and pareto efficient allocation exists
under certain conditions. \citet{TT93} study the structure of envy-free allocations of a single indivisible good when 
monetary compensations are possible. \citet{Mas87} studies a similar model to \cite{Sve83} under slightly different 
conditions; he showed that, with sufficient money, an envy-free allocation always exists. Specifically, 
his results imply that if the agents are unit-demand and their value for each item is at most one dollar, then 
a total of $n-1$ dollars suffices for envy-freeness. In \citet{Ara95} and \citet{Kli00}, the authors 
consider the same model and give polynomial-time algorithms to compute an envy-free allocation with subsidy.

Among the papers that consider a setting with more than $n$ items, most reduce to the above $n$-item case 
where at most one good is allocated to each agent. For example, \citet{ADG91} consider the more general $
m$-item setting and allow the possibility of 
undesirable objects, but their procedure introduces either ``null objects'' or ``fictitious people'' to equalize the 
number of agents and items before outputting an allocation with single-item bundles. \citet{HRS02} also consider the $m$-good case 
and provide a procedure to compute an envy-free allocation with side-payments, but their approach begins by bundling 
the goods into $n$ sets.

In recent work, \citet{HS19} extend the above models to the multi-demand setting with any number~$m$ of indivisible 
goods. Specifically, they consider the setting in which the $n$~agents have {\em additive} valuation functions over a set 
of $m$ items, and, without loss of generality, the value of each item is at most $1$. They characterize 
the \textit{envy-freeable} allocations in terms of the structure of the \textit{envy graph} (see Section~\ref{ss:envyfreeable}), 
whose nodes are the agents and 
whose arc weights represent the envies between pairs of agents. They then study the problem of minimizing the amount of 
subsidy that is sufficient to guarantee envy-freeness. It is easy to see this {\em minimum subsidy} can be at least $n-1$ for 
all envy-freeable allocations. Indeed, consider the case of a single item which each agent values at exactly one dollar;
evidently, every agent that does not receive the item must be compensated with a dollar. They present a matching 
upper bound of $n-1$ dollars for the special 
cases of {\em binary} and {\em identical} additive valuations. 

More generally, they prove that, for additive valuations, an envy-freeable allocation always exists if the total 
subsidy at least is $m\cdot(n-1)$ dollars.
But, based on the experimental analysis of over $100,000$ synthetic instances and over $3,000$ real-world instances 
of fair division, \citet{HS19} conjecture that this upper bound can be improved to $n-1$ dollars.
That is, for agents with additive valuations an envy-freeable allocation that requires a 
subsidy of at most $n-1$ always exists. In addition, they conjecture that an allocation exists that is 
both envy-freeable (with perhaps a much larger subsidy) 
{\em and} EF1 for the fair division problem with additive valuations.  
\begin{conj} \label{conj:hsconj1}
\cite{HS19} For additive valuations, there is an envy-freeable allocation that requires a total subsidy of at most $n-1$ dollars.
\end{conj}
\begin{conj} \label{conj:hsconj2}
\cite{HS19} For additive valuations, there is an envy-freeable allocation that is EF1.
\end{conj}

\subsection{Our Results} \label{ss:ourresults}

In this work we settle both Conjecture~\ref{conj:hsconj1} and Conjecture~\ref{conj:hsconj2}.
In fact, our main result is even stronger in a several ways. To wit, we show that, for any instance with additive valuations, there is an 
allocation that is simultaneously envy-freeable, EF1, balanced, \textit{and} requires a total subsidy of at most $n-1$ dollars. 
Moreover, we present an algorithm that computes such an allocation in polynomial time. Our bound not only 
applies to the total subsidy, but to each individual payment -- the payment made to each agent in this 
allocation is at most one dollar! 
Formally, in Sections~\ref{s:alg} and~\ref{s:bound} we prove the following theorem.
\begin{thm}\label{thm:main}
For additive valuations there is 
an envy-freeable allocation where the subsidy to each agent is at most one dollar.
(This allocation is also EF1, balanced, and can be computed in polynomial time.)
\end{thm}
It is easy to see that, when minimizing the total subsidy, at least one agent will not receive a subsidy.
Thus Theorem~\ref{thm:main} implies that the total subsidy required is indeed at most $n-1$ dollars.

In Section~\ref{s:monotone} we consider the general setting where the agents have arbitrary 
monotone valuation functions. Analogously, without loss of generality, we may scale the valuations so that 
marginal value of each item for any agent never exceeds one dollar. We show that there is an envy-freeable allocation in 
which the subsidy required is at most $2(n-1)$ dollars per agent.
Thus, the total subsidy required to ensure the existence of an envy-free allocation at most $O(n^2)$. 
Note that the assumption of monotonicity is extremely mild and so the valuations the agents have for bundles of items 
may range from $0$ to $\Omega(m)$ in quite an arbitrary manner. Consequently, it is somewhat remarkable that
the total subsidy required to ensure the existence of an envy-free allocation is independent of the number of items $m$.
In particular, when $m$ is large the subsidy required is negligible in terms of $m$ and thus, typically, also negligible 
in terms of the values of the allocated bundles.
In this case, given a valuation oracle for each agents, the corresponding envy-free allocation and subsidies
can be computed in polynomial time. Specifically, in Section~\ref{s:monotone} we prove:

\begin{thm}\label{thm:monotone}
For monotonic valuations there is 
an envy-freeable allocation where the subsidy to each agent is at most $2(n-1)$ dollars. (Given a valuation oracle, 
this allocation can be computed in polynomial time.)
\end{thm}

In effect, our work implies that there is, in fact, a much stronger connection between the classical divisible goods (cake-cutting) setting and the indivisible goods setting than was previously known. While the classical guarantees (envy-freeness and proportionality) can be achieved with divisible goods, for the indivisible-goods setting much of the recent literature focuses on achieving weaker fairness properties. We show that by simply introducing a small subsidy that only depends on the number of agents, the much stronger classical guarantees \textit{can} be achieved in the indivisible goods setting. Moreover, allocations that give these classical guarantees with a small bounded subsidy can be efficiently found.

\section{The Fair Division with Subsidy Problem} \label{s:Preliminaries}

There is a set $I=\{1,2,\dots, n\}$ of agents and a set $J=\{1,2,\dots, m\}$ of indivisible goods (items). 
Each agent $i\in I$ has a valuation function $v_i$ over the set of items. That is,
for each bundle $S\subseteq J$ of items, agent $i$ has value $v_i(S)$. We make the standard assumptions that the valuation 
functions are {\em monotonic}, 
that is, $v_i(S)\leq v_i(T)$ when $S\subseteq T$, and that $v_i(\emptyset) = 0$. An agent $i$ and valuation function $v_i$ 
are \textit{additive} if, for each item $j\in J$, agent $i$ has value $v_i(j)=v_i(\{j\})$, and for any collection $S\subseteq J$, 
agent $i$ has value $v_i(S)=\sum_{j\in S} v_i(j)$. We denote the vector of valuation functions by $\mbs{v}=(v_1,\ldots,v_n)$, 
and call $\mbs{v}$ a {\em valuation profile}. Additionally, without loss of generality we scale each agent $i$'s valuation function 
so that the maximum marginal value of any item $j$ is at most $1$. Specifically, for additive valuations, 
this implies $v_i(j) \leq 1$ for every agent $i$ and item $j$.

An {\em allocation} is an ordered partition $\mathcal{A} = \{A_1,\ldots,A_n\}$ of the set of items into $n$ bundles.
Agent~$i$ receives the (possibly empty) bundle $A_i$ in the allocation $\mathcal{A}$. The allocation $\mathcal{A}$ is \textit{envy-free} if
\begin{align*}
    v_i(A_i) \ \geq\  v_i(A_{k}) \quad\quad\quad \forall i\in I, \forall k\in I.
\end{align*}
That is, for any pair of agents $i$ and $k$, agent $i$ prefers its own bundle $A_{i}$ 
over the bundle $A_k$. In the (envy-free) {\em fair division problem} the objective is to find
an envy-free allocation of the items.

Unfortunately, this objective is generally impossible to satisfy. A natural relaxation of the objective arises by incorporating
subsidies. Specifically, let $\mbs{p}=(p_1,\ldots,p_n)$ be a non-negative subsidy vector, where agent $i$ 
receives a payment $p_i\ge 0$. An \textit{allocation with payments} $(\mathcal{A},\mbs{p})$ is then envy-free if
\begin{align*}
v_i(A_i) +p_i  \geq v_i(A_{k}) +p_k \quad\quad\quad \forall i\in I, \forall k\in I.
\end{align*}
That is, each agent prefers its bundle plus payment over the bundle plus payment of every other agent.
In the {\em fair division with subsidy problem} the objective is to find an envy-free allocation with payments
whose total subsidy $\sum_{i\in I} p_i$ is minimized.

\subsection{Envy-Freeability and the Envy Graph} \label{ss:envyfreeable}
For any fixed allocation $\mathcal{A}$, a payment vector $\mbs{p}$ such that $\{\mathcal{A}, \mbs{p}\}$ is
envy-free does not always exist. To see this, consider an instance with a single item and agents $I=\{1,2\}$ 
with values $v_1<v_2$ for the item. Now take the fixed allocation where the item is given to agent~$1$.
It follows that agent~$2$ must receive a payment of at least $v_2$ to eliminate its envy. But then, because $v_2 > v_1$, agent~$1$ is 
envious of the bundle plus payment allocated to agent~$2$. Thus, no payment vector can eliminate the envy of both agents for 
this allocation.

We call an allocation $\mathcal{A}$ \textit{envy-freeable} if there 
exists a payment vector $\mbs{p} = (p_1,\ldots,p_n)$ such that $\{\mathcal{A},\mbs{p}\}$ is envy-free. 
There is a nice graphical characterization for the envy-freeability of an allocation $\mathcal{A}$.
The {\em envy graph}, denoted $G_{\mathcal{A}}$, for an allocation $\mathcal{A}$ is a 
complete directed graph with vertex set $I$. For any pair of agents $i,k\in I$ the weight of arc $(i,k)$ in $G_{\mathcal{A}}$
is the envy agent $i$ has for agent $k$ under the allocation $\mathcal{A}$, that is, $w_{\mathcal{A}}(i,k) \ =\  v_i(A_{k}) - v_i(A_i)$.

An allocation is envy-freeable if and only if its envy graph does not contain a positive-weight directed cycle.
More generally, \citet{HS19} obtained the following theorem; we include their proof in order to 
familiarize the reader with the structure of envy-freeable allocations.
\begin{thm} \label{thm:hsthm1}
\cite{HS19} The following statements are equivalent.
\vspace{-.25cm}
\begin{enumerate}[noitemsep]
\item[(a)] The allocation $\mathcal{A}$ is envy-freeable.
\item[(b)] The allocation $\mathcal{A}$ maximizes (utilitarian) welfare across all reassignments of its bundles to agents: for 
    every permutation $\pi$ of $I=[n]$, we have $\sum_{i\in I}v_i(A_i) \geq \sum_{i\in I}v_i(A_{\pi(i)})$.
 \item[(c)] The envy graph $G_{\mathcal{A}}$ contains no positive-weight directed cycles.
\end{enumerate}
\end{thm}
\begin{proof}\ \\
$(a) \Rightarrow (b)$: Let $\mathcal{A} = \{A_1,\ldots,A_n\}$ be envy-freeable. Then, by definition, there exists a payment 
vector $\mbs{p}$ such that $v_i(A_i) + p_i \geq v_i(A_{k}) + p_{k}$, for any pair of agents $i$ and $k$.
Rearranging, we have $v_i(A_{k})-v_i(A_i) \leq p_i - p_{k}$. Then, for any permutation $\pi$ of $I=[n]$
\begin{align*}
\sum_{i\in I} \left(v_i(A_{\pi(i)}) - v_i(A_i)\right) \ \leq\  \sum_{i\in I} \left(p_i - p_{\pi(i)}\right) \ = \ \sum_{i\in I} p_i - \sum_{i\in I} p_{\pi(i)} \ = \ 0.
\end{align*}
Thus the allocation $\mathcal{A}$ maximizes welfare over all reassignments of its bundles.\\
$(b) \Rightarrow (c)$: 
Assume $\mathcal{A}$ maximizes welfare over all reassignments of its bundles and take a
directed cycle $C$ in the envy graph $G_{\mathcal{A}}$. 
Without loss of generality $C=\{1,2,\dots,r\}$ for some $r\ge 2$. Now define
a permutation $\pi_C$ of $I$ according to the following rules: (i) $\pi_C(i) = i+1$ for each $i\le r-1$, (ii) $\pi_C(r) = 1$, and (iii) $\pi_C(i) = i$ otherwise.
Then the weight of the cycle $C$ in the envy graph satisfies
\begin{align*}
w_{\mathcal{A}}(C) &= \sum_{(i,k)\in C} w_{\mathcal{A}}(i,k) \\
 &= \sum_{i=1}^{r-1} \left(v_i(A_{i+1}) - v_i(A_i)\right) + \left( v_r(A_{1}) - v_r(A_r)\right)\\
  &= \sum_{i=1}^{r-1} \left(v_i(A_{i+1}) - v_i(A_i)\right) + \left( v_r(A_{1}) - v_r(A_r)\right) + \sum_{i=r+1}^{n} \left(v_i(A_{i}) - v_i(A_i)\right) \\
    &= \sum_{i\in I}v_i(A_{\pi(i)})-v_i(A_i) \\
    & \leq 0.
\end{align*}
The inequality holds as $\mathcal{A}$ maximizes welfare over all bundle reassignments. Thus $C$ has non-positive weight.\\
$(c) \Rightarrow (a)$: Assume the envy graph $G_{\mathcal{A}}$ contains no positive-weight directed cycles. Let $\ell_{G_{\mathcal{A}}}(i)$ be the 
maximum weight of any path (including the empty path) that starts at vertex $i$ in $G_{\mathcal{A}}$.
For each agent $i\in I$, set its payment $p_i=\ell_{G_{\mathcal{A}}}(i)$. Observe that $p_i\ge 0$
as the empty path has weight zero.
The corresponding pair $(\mathcal{A},\mbs{p})$ is then envy-free. To see this, recall that there are no positive-weight cycles.
Therefore, for any pair of agents $i$ and $k$, we have 
\begin{align*}
p_i \ =\  \ell_{G_{\mathcal{A}}}(i) \ \geq\  w_{\mathcal{A}}(i,k) + \ell_{G_{\mathcal{A}}}(k)  \ =\  \left( v_i(A_{k}) - v_i(A_i)\right)+ p_{k}.
\end{align*}
Thus $v_i(A_i) +p_i  \geq v_i(A_{k}) +p_k$ and the allocation $\mathcal{A}$ is envy-freeable.
\end{proof}
Theorem~\ref{thm:hsthm1} is important for two reasons.
First, whilst an allocation $\mathcal{A}=\{A_1,A_2, \dots,A_n\}$ need not be envy-freeable, Condition~(b) tells us that
there is some permutation $\pi$ of the bundles in $\mathcal{A}$ such that the resultant allocation,
$\mathcal{A}^{\pi}=\{A_{\pi(1)},A_{\pi(2)}, \dots, A_{\pi(n)}\}$, is envy-freeable! 
For example, consider again the simple one-item, two-agent instance above. If the item is allocated to agent~$1$
then the weight on the arc $(1,2)$ is~$-v_1$ and the weight on the arc $(2,1)$ is~$v_2$. Because $v_1 < v_2$, 
the envy graph has a positive-weight directed cycle $\{1,2\}$ and so, by Theorem~\ref{thm:hsthm1}, this allocation is not 
envy-freeable. However, suppose we fix the bundles and find a utility-maximizing reallocation of these fixed bundles. 
This reallocation assigns the item to agent~$2$ and now there is no positive-weight directed 
cycle in the resultant envy-free graph; consequently this allocation is envy-freeable by providing a subsidy in the range $[v_1,v_2]$ 
to agent~$1$.

Second, to calculate the subsidy vector $\mbs{p}$ associated with an envy-freeable allocation, such as $\mathcal{A}^{\pi}$,
it suffices to calculate the maximum-weight paths beginning at each vertex in its envy graph. (In fact, it is straightforward to prove that
the heaviest-path weights \textit{lower bound} the payment to each agent in any envy-free payment vector of an envy-freeable
allocation~\cite{HS19}.) Note that given any payment vector that eliminates envy, we may uniformly increase or 
decrease the payments to all agents while maintaining envy-freeness. As a consequence, in the payment vector 
that minimizes the total subsidy, there is at least one agent that receives a payment of $0$.
Together these arguments give the following very useful observation.
\begin{observation}\label{obs:useful}
For any envy-freeable allocation $\mathcal{A}$, the minimum total subsidy required is at most~$(n-1)\cdot \ell_{G_\mathcal{A}}^{\max}$, where
$\ell_{G_\mathcal{A}}^{\max}$ is the maximum weight of a directed path in the envy graph $G_\mathcal{A}$.
\end{observation}

\citet{HS19} then prove:
\begin{thm} \label{thm:hsthmub}\cite{HS19} 
For any envy-freeable allocation $\mathcal{A}$, the minimum total subsidy required is at most~$(n-1)\cdot m$.
\end{thm}
\begin{proof}
In a minimum subsidy vector, at least one agent requires no subsidy.
Thus it suffices to show that the subsidy to any agent~$i$ is at most $m$. By Observation~\ref{obs:useful},
it suffices to show that the heaviest path weight starting at any vertex is at most $m$.
Without loss of generality, let the heaviest path be $P=\{1,2,\dots,r\}$.
The subsidy made to agent~$1$ can then be upper bounded by
     \begin{equation*}
\ell_{G_{\mathcal{A}}}(1) 
\  = \  \sum_{(i,k)\in P} w_{\mathcal{A}}(i,k) 
 \ =\  \sum_{i=1}^{r-1} \left(v_i(A_{i+1}) - v_i(A_i)\right)
 \   \le\   \sum_{i=1}^{r-1} v_i(A_{i+1}) 
 \   \le \  \sum_{i=1}^{r-1} |A_{i+1}| 
    \  \le\   |J| 
    \   =\  m.
    \end{equation*}
    
    Here the second inequality holds because each agent has value at most one for any item. The third
    inequality is due to the fact that for the allocation $\mathcal{A}$ the bundles $\{A_1, A_2,\dots, A_n\}$ are disjoint.
    Consequently $p_i\le m$ for each agent, as required.
\end{proof}

For an arbitrary envy-freeable allocation $\mathcal{A}$ the bound in Theorem~\ref{thm:hsthmub} is tight. 
To see this, consider the example where every agent has value $1$ for each item, and the grand 
bundle (containing all items) is given to agent~$1$. This allocation 
is envy-freeable, and here each of the other $n-1$ agents requires a subsidy of $m$ for envy-freeness.
Ergo, to provide an improved bound on the total subsidy, we cannot consider any generic envy-freeable allocation.
Instead, our task is find a specific envy-freeable allocation where the heaviest paths in the associated envy graph 
have much smaller weight. In particular, for the case of additive agents, we want that these path weights are at most $1$ 
rather than at most $m$. This is our goal in the subsequent sections of the paper.

Before doing this, let us briefly discuss some computational aspects. Theorem~\ref{thm:hsthm1} provides efficient methods to test if a 
given allocation is envy-freeable. For example, this can be 
achieved via a maximum-weight bipartite matching algorithm to verify Condition~(b). Alternatively, 
Condition~(c) can be tested in polynomial time using the Floyd-Warshall algorithm.\footnote{In fact, a simple reduction 
converts the problem of finding minimum payments for a fixed allocation into a shortest-paths problem and any 
efficient shortest-paths algorithm can be applied.}
Finally, given an arbitrary non-envy-freeable allocation $\mathcal{A}$, one can efficiently find a corresponding 
envy-freeable allocation $\mathcal{A}^{\pi}$ by fixing the $n$ bundles of the given allocation 
and computing a maximum-weight bipartite matching between the agents and the bundles.

\section{An Allocation Algorithm for Additive Agents} \label{s:alg}
In this section we present an allocation algorithm for the case of additive agents. 
Recall our task is to construct an envy-freeable allocation $\mathcal{A}$ 
with maximum path weight $1$ in the envy graph $G_{\mathcal{A}}$.
We do this via an allocation algorithm defined on the valuation graph for the
instance. The valuation graph $H$ is the complete bipartite graph on vertex sets $I$ and $J$, where edge $(i,j)$ has weight $v_i(j)$.
We denote by $h[\hat{I}, \hat{J}]$ the subgraph of $H$ induced by $\hat{I}\subseteq I$ and $\hat{J}\subseteq J$.
The allocation algorithm then proceeds in rounds where each agent is matched to exactly one item in each round. 
For the first round, we set $J_1=J$. In round $t$, we then find a maximum-weight matching $M_t$ in $H[I,J_t]$. If agent $i$ is matched to item
$j=\mu^t_i$ then we allocate item $\mu^t_i$ to that agent. We then recurse on the remaining items $J_{t+1}=J_t\setminus \cup_{i\in I} \mu^t_i$.
The process ends when every item has been allocated.
This procedure is formalized via pseudocode in Algorithm~\ref{alg:allocation}.
 \begin{algorithm}[ht]
 \SetAlgoLined
  $A_i\gets \emptyset$ for all $i\in I$\;
  $t\gets1; J_1\gets J$\;
  \While{$J_t\neq \emptyset$}{
   Compute a \textit{maximum-weight matching} $M^t=\{(i, \mu^t_i)\}_{i\in I}$ in $H[I, J_t]$\;
   Set $A_i\gets A_i\cup \{\mu^t_i\}$ for all $i\in I$\;
   Set $J_{t+1}\gets J_t\setminus \cup_{i\in I} \mu^t_i$\;
   $t\gets t+1$\;
   }
  \caption{Bounded-Subsidy Algorithm}\label{alg:allocation}
 \end{algorithm}

Suppose the algorithm terminates in $T$ rounds. We assume that every agent receives an item in each round. 
For rounds $1$ to $T-1$ this is evident because agent $i$ can be assigned a item for which it has zero value.
For round $T$, we assume there are exactly $n$ items remaining, possibly by adding dummy items of 
no value to any agent. 

This algorithm has many interesting properties. In this section we prove that it outputs an envy-freeable allocation $\mathcal{A}$.
Furthermore, the allocation $\mathcal{A}$ is EF1, thus settling Conjecture~\ref{conj:hsconj2}.
The allocation is also balanced in that (discarding any additional dummy items)
the bundles that the agents receive differ in size by at most one item; in particular, each agent receives a bundle 
of size either $\lfloor \frac{m}{n}\rfloor$ or $\lceil \frac{m}{n}\rceil$. The allocation algorithm also clearly runs in polynomial time.

We also show in this section that any allocation $\mathcal{A}$ that is both envy-freeable and EF1 has a heaviest path weight in the
envy graph of weight at most $n-1$. Thus, By Observation~\ref{obs:useful}, the algorithm outputs an allocation that requires a 
subsidy of at most $(n-1)^2$. 
As claimed though, the heaviest path weight in $G_\mathcal{A}$ is in fact at most one and so the total subsidy
needed is at most $n-1$. We defer the proof of this fact, our main result, to Section~\ref{s:bound}.

\subsection{The Allocation Is Envy-freeable}
Let's first see that the output allocation $\mathcal{A}$ is envy-freeable.
\begin{lem}\label{lem:envyfreeable}
The output allocation $\mathcal{A}$ is envy-freeable.
\end{lem}
\begin{proof}
    Let $M^t$ be the maximum matching found in round $t$ and $\mbs{\mu}^t=\{\mu_1^t,\mu^t_2, \ldots, \mu_n^t\}$ the corresponding 
    items allocated in that round. By Theorem~\ref{thm:hsthm1} it suffices to show that no directed cycle in the 
    envy graph corresponding to the final allocation $\mathcal{A}$ has positive weight. Take any directed cycle $C$ in the 
    envy graph $G_{\mathcal{A}}$. Again, we may assume without loss of generality that $C=\{1,2,\dots,r\}$ for some $r\ge 2$.
    We have
    \begin{align*}
        w_{\mathcal{A}}(C) &= \sum_{(i,k)\in C} w_{\mathcal{A}}(i,k) &\\
        &= \sum_{(i,k)\in C}\left[v_i(A_{k}) - v_i(A_i)\right] &\\
        &= \sum_{(i,k)\in C} \sum_{t=1}^T \left[v_i(\mu_{k}^t) - v_i(\mu_i^t)\right] &\\
        &= \sum_{(i,k)\in C}  \sum_{t=1}^T w_{\mbs{\mu}^t}(i,k) &\\
        &=  \sum_{t=1}^T \sum_{(i,k)\in C} w_{\mbs{\mu}^t}(i,k).
    \end{align*}
    Let $\pi_C$ be the permutation of $I$ under which $\pi_C(i) = i+1$ for each $i\le r-1$, $\pi_C(r) = 1$, and $\pi_C(i) = i$ otherwise.
    In each round $t$, since $M_t$ is a maximum-weight matching, $\sum_{(i,k)\in C} w_{\mbs{\mu}^t}(i,k)$ is non-positive: 
    otherwise, the matching $\hat{M}^t$ obtained by allocating to each agent $i$ the item $\mu_{\pi_C(i)}^t$ has greater weight 
    than $M_t$, a contradiction. Thus $w_{\mathcal{A}}(C)$ is also non-positive. Consequently, by 
    Theorem~\ref{thm:hsthm1} the allocation produced by the algorithm is envy-freeable.
\end{proof}

\subsection{The Allocation Is EF1}
We say that an allocation $\mathcal{A}$ satisfies the \textit{envy bounded by a single good} 
property, and is \textit{EF1}, if for each pair $i,k$ of agents, either
$A_k = \emptyset$ or there exists an item $j\in A_{k}$ 
such that
\begin{align*}
    v_i(A_i) \geq v_i(A_{k}\setminus\{j\}).
\end{align*}
Next, let's prove the output allocation $\mathcal{A}$ is EF1. 
\begin{lem}\label{lem:ef1}
 The output allocation $\mathcal{A}$ is EF1.
\end{lem}
\begin{proof}
    Let $\mathcal{A} = \{A_1,\ldots,A_n\}$. Recall, in any round $t$, the algorithm 
    computes a maximum-weight matching $M^t$ in $H[I, J_t]$ and allocates item $\mu_i^t$ to agent $i$.
    Thus $A_i=\{\mu_i^1,\ldots,\mu_i^T\}$ is the set of 
    items allocated to agent $i$. Observe that
    $v_i(\mu_i^t) \geq v_i(j)$ for any item $j\in J_{t+1}$, the collection of items unallocated at the start of round $t+1$. 
    Otherwise, we can replace the edge $(i,\mu_i^t)$ with $(i,j)$ in $M_t$, to obtain a higher-weight matching in $H[I, J_t]$. 
    Therefore, for any pair of agents $i$ and 
    $k$, we have
    \begin{align*}
        v_i(A_i) &= v_i(\{\mu_i^1,\ldots,\mu_i^T\}) \\
        &= v_i(\mu_i^1) + \cdots + v_i(\mu_i^{T-1}) + v_i(\mu_i^T) \\
        &\ge v_i(\mu_i^1) + \cdots + v_i(\mu_i^{T-1})\\
        &\geq v_i(\mu_{k}^2) + \cdots + v_i(\mu_{k}^T) \\
        &= v_i(A_{k}\setminus\{\mu_{k}^1\}).
    \end{align*}
Ergo, the output allocation $\mathcal{A}$ is EF1.
\end{proof}

\begin{clm}\label{clm:envy-freeable-EF1}
Let $\mathcal{A}$ be both envy-freeable and EF1. Then the minimum total subsidy required is at most~$(n-1)^2$.
\end{clm}
\begin{proof}
Since there is an agent that requires no subsidy, it suffices to prove that the maximum path 
weight in the envy graph $G_\mathcal{A}$ is at most $n-1$.
But $\mathcal{A}$ is EF1. So agent $i$ envies agent $k$ by at most one, the maximum value of a single item. 
Thus every arc $(i,k)$ has weight at most one, that is, $w_{\mathcal{A}}(i,k) \leq 1$. The result follows as
any path contains at most $n-1$ arcs.
\end{proof}

Since we have shown that the output allocation $\mathcal{A}$ is both envy-freeable and EF1, 
it immediately follows by Claim~\ref{clm:envy-freeable-EF1} that it requires a total subsidy 
of at most $(n-1)^2$.

\section{The Subsidy Required Is at Most One per Agent} \label{s:bound}
In this section we complete our analysis of the additive setting. By the EF1 property of the output 
allocation $G_\mathcal{A}$ we have an upper bound of~$1$
on the weight of any arc in the envy graph $G_\mathcal{A}$. 
But this is insufficient to accomplish our goal of proving that the envy graph has maximum path weight~$1$.
How can we do this? As a thought experiment, imagine that, rather than an upper bound of~$1$ on each 
arc weight, we have a lower bound of~$-1$ on each arc weight. The subsequent lemma proves this would be a sufficient condition!

\begin{lem}\label{thm:one-dollar-less}
Let $\mathcal{A}$ be an envy-freeable allocation. If $w_{\mathcal{A}}(i,k) \ge -1$ for every arc $(i,k)$ in 
the envy graph then the maximum subsidy required is at most one per agent.
\end{lem}
\begin{proof}
By Theorem~\ref{thm:hsthm1}, as $\mathcal{A}$ is an envy-freeable the envy graph $G_\mathcal{A}$
contains no positive-weight cycles. Let $P$ be the maximum-weight path in $G_\mathcal{A}$.
Without loss of generality, $P=\{1,2,\dots, i\}$ with weight $p_1=\ell_{G_{\mathcal{A}}}(1)$.
Now take the directed cycle $C=P\cup (i,1)$. Because $C$ has non-positive weight and every arc weight is at least $-1$, 
we obtain
\begin{equation*}
0 \ \ge \ w_{\mathcal{A}}(C) \ =\  \ell_{G_{\mathcal{A}}}(1) + w_{\mathcal{A}}(i,1) \ \ge\ \ell_{G_{\mathcal{A}}}(1)-1.
\end{equation*}
Therefore  $\ell_{G_{\mathcal{A}}}(1)\le 1$ and the maximum subsidy is at most one.
\end{proof}

At first glance, Lemma~\ref{thm:one-dollar-less} seems of little use. We already know 
every arc in the envy graph has weight at most $1$. Suppose in addition that every arc weight was 
at least $-1$. That is, $1\ge w_{\mathcal{A}}(i,k) \ge -1$ for each arc $(i,k)$.
Consequently, $v_i(A_i) \le v_i(A_{k})+1$ and $v_i(A_k) \le v_i(A_{i})+1$.
In instances with a large number of valuable items this means that every agent is essentially indifferent over which
bundle in $\mathcal{A}$ they receive. It is unlikely that an allocation with this property even exists for 
every instance, and certainly not the case that our algorithm outputs such 
an allocation.

The trick is to apply Lemma~\ref{thm:one-dollar-less} to a modified fair division instance. 
In particular we construct, for each agent $i$, a modified valuation function $\bar{v}_i$ from $v_i$.
We then prove that the allocation $\mathcal{A}^{\mbs{v}}$ output for the original valuation profile
$\mbs{v}$ is envy-freeable even for the modified valuation profile $\mbs{\bar{v}}$.
Next we show that with this same allocation, every arc weight is at least $-1$ in the envy graph under the modified 
valuation profile $\mbs{\bar{v}}$. By Lemma~\ref{thm:one-dollar-less}, this implies that
the maximum subsidy required is at most one for the 
valuation profile $\mbs{\bar{v}}$. To complete the proof we show that the maximum subsidy 
required by each agent for the original valuation profile $\mbs{v}$ is at most the subsidy required for $\mbs{\bar{v}}$.

\subsection{A Modified Valuation Function}
Let $\mathcal{A}^{\mbs{v}} = \{A^{\mbs{v}}_1,\ldots,A^{\mbs{v}}_n\}$ be the allocation output by our algorithm under the original
valuation profile $\mbs{v}$. We now create the modified valuation profile $\mbs{\bar{v}}$. 
For each agent $i$, define $\bar{v}_i$ according to the rule:
\begin{align*}
    \bar{v}_i(\mu_i^t) &= v_i(\mu_i^t) &\forall t\le T\\
    \bar{v}_i(\mu_{k}^t) &= \max\left(v_i(\mu_{k}^t),v_i(\mu_i^{t+1})\right)  &\forall k\in I\setminus\{i\},\  \forall t\le T-1 \\
    \bar{v}_i(\mu_{k}^T) &= v_i(\mu_{k}^T) &\forall k\in I\setminus\{i\}.
    \end{align*}
That is, the value $\bar{v}_i(j)$ remains the same for any item $j\in A^{\mbs{v}}_i$ that was allocated to agent $i$ by the algorithm. 
For any other item $j$, the value $\bar{v}_i(j)$ is the maximum of the original value $v_i(j)$ and the value of the item 
allocated to $i$ by the algorithm in the round that immediately follows the round where $j$ was allocated to some agent.

The following two observations are trivial but will be useful.
\begin{observation}\label{obs:equal}
For any agent $i$ and item $j\in A^{\mbs{v}}_i$, we have $v_i(j) = \bar{v}_i(j)$. \qed
\end{observation}
\begin{observation}\label{obs:upper}
For any agent $i$ and item $j\notin A^{\mbs{v}}_i$, we have $v_i(j) \le \bar{v}_i(j)$. \qed
\end{observation}

We will show the bound on the subsidy by a sequence of claims based on the proof plan outlined above.
First we show that $\mathcal{A}^{\mbs{v}}$ envy-freeable even under the modified valuation profile.
\begin{clm}\label{clm:allocation}
The allocation $\mathcal{A}^{\mbs{v}}$ output under the original valuation profile $\mbs{v}$ is an envy-freeable 
allocation under the modified valuation profile $\mbs{\bar{v}}$.
\end{clm}
\begin{proof}
By Theorem~\ref{thm:hsthm1}, to show that the allocation $\mathcal{A}^{\mbs{v}}$ is envy-freeable under the modified valuation profile $\mbs{\bar{v}}$ 
we must show that there is no positive-weight cycle in the envy graph using the modified values.
So suppose cycle $C$ has positive modified weight.
To obtain a contradiction, first observe that, in the allocation $\mathcal{A}^{\mbs{v}}$, 
agent $i$ receives the bundle $\mathcal{A}_i^{\mbs{v}} =\{\mu^1_i,\mu^2_i,\dots, \mu^T_i\}$.
Thus with respect to $\mbs{\bar{v}}$ the envy agent $i$ has for agent $k$ is
\begin{equation}\label{eq:modified-envy}
\bar{v}_i(\mathcal{A}_k^{\mbs{v}}) -\bar{v}_i(\mathcal{A}_i^{\mbs{v}}) 
\ = \ \sum_{t=1}^T \bar{v}_i(\mu^t_k) -  \sum_{t=1}^T \bar{v}_i(\mu^t_i) 
\ = \ \sum_{t=1}^T  \left( \bar{v}_i(\mu^t_k) -  \bar{v}_i(\mu^t_i) \right). 
 \end{equation}
 As the envy graph contains a positive-weight cycle $C$ we have, by (\ref{eq:modified-envy}), that
 \begin{align*}
0 < \sum_{(i,k)\in C} \bar{v}_i(\mathcal{A}_k^{\mbs{v}}) -\bar{v}_i(\mathcal{A}_i^{\mbs{v}}) 
\ = \ \sum_{(i,k)\in C}  \sum_{t=1}^T  \left( \bar{v}_i(\mu^t_k) -  \bar{v}_i(\mu^t_i) \right) 
\ = \  \sum_{t=1}^T \sum_{(i,k)\in C}  \left( \bar{v}_i(\mu^t_k) -  \bar{v}_i(\mu^t_i) \right). 
 \end{align*}
 This implies there exists a round $t$ such that
 \begin{equation}\label{eq:bad-cycle}
 \sum_{(i,k)\in C}  \bar{v}_i(\mu^t_k)  >   \sum_{(i,k)\in C} \bar{v}_i(\mu^t_i). 
 \end{equation}
 Now $M^t$ is a maximum-weight matching in $H[I, J_t]$ for the original valuation profile $\mbs{v}$.
 Let $\hat{M}^t$ be the matching formed from $M^t$ by permuting around the cycle $C$ the bundles of the agents in $C$.
 But then, by (\ref{eq:bad-cycle}), the matching $\hat{M}^t$ has greater weight in $H[I, J_t]$ than the matching $M^t$
 for the modified profile $\bar{\mbs{v}}$.
Consequently, we will obtain our contradiction if we can prove that $M^t$ is a maximum-weight matching in $H[I, J_t]$
even with respect to~$\mbs{\bar{v}}$.

This is true in the final round matching; clearly $M^T$ is a maximum-weight matching in 
$H[I, J_T]$ because, by definition, $\mbs{\bar{v}}$ and $\mbs{v}$ have the same value
for items in $J_T$.  Thus, it remains to prove the statement for each round $t \le T-1$.
Now $\mbs{\mu}^t=\{\mu_1^t, \mu_2^t,\ldots, \mu_n^t\}$ is the allocation of the items round $t$. Again, 
for a contradiction, assume that matching $M^t$ is not maximum in $H[I, J_t]$ for 
the valuation profile $\mbs{\bar{v}}$. Then, by Theorem~\ref{thm:hsthm1}, 
the envy graph $G_{\mbs{\mu^t}}$ contains a positive-weight directed cycle $C$. Without loss of generality, let $C = \{1,\ldots,r\}$.

We divide our analysis into two cases, depending on whether the weights on the arcs of $C$ change when the valuation 
profile is modified from $\mbs{v}$ to $\mbs{\bar{v}}$. Specifically, we call an arc $(i,i+1)$ of $C$ \textit{blue} 
if $\bar{v}_{i}(\mu_{i+1}^t) = v_{i}(\mu_{i+1}^t)$, that is, agent $i$'s value for the item allocated to agent $i+1$ 
does not change when the valuation profile is modified. We call an arc \textit{red} otherwise. Observe that if the
 arc $(i,i+1)$ of $C$ is red, then $\bar{v}_{i}(\mu_{i+1}^t) = v_{i}(\mu_{i}^{t+1}) > v_{i}(\mu_{i+1}^t)$, so in the 
 original valuation function agent $i$ strictly prefers the item that it is allocated in round $t+1$ to the item that agent $i+1$ is 
 allocated in round $t$. In turn, this implies that the weight on any red arc is necessarily negative. We have the following two cases to consider.

\begin{itemize}
    \item[(i)] \textit{Every arc of $C$ is blue.} Let $\pi_C$ be the permutation of $I$ under which $\pi_C(i) = i+1$ for 
    each $i\le r-1$, $\pi_C(r) = 1$, and $\pi_C(i) = i$ otherwise. The 
    matching $\mathcal{M}^t$ obtained by allocating to each agent $i$ the item $\mu_{\pi_C(i)}^t$ has greater weight than $M^t$ with 
    respect to the original valuation profile $\mbs{v}$, contradicting the assumption that the algorithm selected a matching of 
    maximum weight.
    \item[(ii)] \textit{$C$ contains a red arc.} In this case, $C$ can be decomposed into a sequence of $d$ directed 
    paths $P_1,\ldots,P_d$ such that each directed path consists of a (possibly empty) sequence of blue arcs followed 
    by exactly one red arc. Figure~\ref{fig:decomp} shows an example of such a decomposition. In the figure, blue arcs 
    are represented by solid lines and red arcs by dashed lines.
    \begin{figure}[ht]
        \centering
        \begin{tikzpicture}[scale=0.5]
            \begin{scope}[every node/.style={circle, fill=none, draw, inner sep=0pt,
            minimum size = 0.15cm
            }]
                
                \node[] (u0) at (0,0) {};
                \node[] (u1) at (2,1) {};
                \node[] (u2) at (3,3) {};
                \node[] (u3) at (2,5) {};
                \node[] (u4) at (0,6) {};
                \node[] (u5) at (-2,5) {};
                \node[] (u6) at (-3,3) {};
                \node[] (u7) at (-2,1) {};
                
                \node[] (v00) at (13,-1) {};
                \node[] (v10) at (15,0) {};
                \node[] (v20) at (16,2) {};
                
                \node[] (v21) at (15.5,4) {};
                \node[] (v31) at (14.5,6) {};
                \node[] (v41) at (12.5,7) {};
                \node[] (v51) at (10.5,6) {};
                
                \node[] (v52) at (8.5,5) {};
                \node[] (v62) at (7.5,3) {};
                \node[] (v72) at (8.5,1) {};
                
                \node[] (v73) at (9.5,0) {};
                \node[] (v03) at (11.5,-1) {};
            
            \end{scope}
    
            \begin{scope}[every edge/.style={draw=blue}]
                
                \path[thick,-{Latex[length=2mm,width=2mm]}] (u0) edge node[label={[label distance=-5]-75:$0.2$}] {} (u1);
                \path[thick,dashed,-{Latex[length=2mm,width=2mm]}] (u1) edge[draw=red] node[label={[label distance=-7]-15:$-0.3$}] {} (u2);
                \path[thick,-{Latex[length=2mm,width=2mm]}] (u2) edge node[label={[label distance=-7]15:$-0.1$}] {} (u3);
                \path[thick,-{Latex[length=2mm,width=2mm]}] (u3) edge node[label={[label distance=-5]75:$0.6$}] {} (u4);
                \path[thick,dashed,-{Latex[length=2mm,width=2mm]}] (u4) edge[draw=red] node[label={[label distance=-5]105:$-0.2$}] {} (u5);
                \path[thick,-{Latex[length=2mm,width=2mm]}] (u5) edge node[label={[label distance=-5]165:$0.4$}] {} (u6);
                \path[thick,dashed,-{Latex[length=2mm,width=2mm]}] (u6) edge[draw=red] node[label={[label distance=-5]195:$-0.4$}] {} (u7);
                \path[thick,dashed,-{Latex[length=2mm,width=2mm]}] (u7) edge[draw=red] node[label={[label distance=-5]255:$-0.1$}] {} (u0);
                
                \path[thick,-{Latex[length=2mm,width=2mm]}] (v00) edge node[label={[label distance=-5]-75:$0.2$}] {} (v10);
                \path[thick,dashed,-{Latex[length=2mm,width=2mm]}] (v10) edge[draw=red] node[label={[label distance=-7]-15:$-0.3$}] {} (v20);
                \path[thick,-{Latex[length=2mm,width=2mm]}] (v21) edge node[label={[label distance=-7]15:$-0.1$}] {} (v31);
                \path[thick,-{Latex[length=2mm,width=2mm]}] (v31) edge node[label={[label distance=-5]75:$0.6$}] {} (v41);
                \path[thick,dashed,-{Latex[length=2mm,width=2mm]}] (v41) edge[draw=red] node[label={[label distance=-5]105:$-0.2$}] {} (v51);
                \path[thick,-{Latex[length=2mm,width=2mm]}] (v52) edge node[label={[label distance=-5]165:$0.4$}] {} (v62);
                \path[thick,dashed,-{Latex[length=2mm,width=2mm]}] (v62) edge[draw=red] node[label={[label distance=-5]195:$-0.4$}] {} (v72);
                \path[thick,dashed,-{Latex[length=2mm,width=2mm]}] (v73) edge[draw=red] node[label={[label distance=-5]255:$-0.1$}] {} (v03);
                
            \end{scope}
    
            \begin{scope}[every node/.style={draw=none,rectangle}]
                
                \node (Clabel) at (0,3) {$C$};
                \node (P1label) at (14,1) {$P_1$};
                \node (P2label) at (12.5,5.5) {$P_2$};
                \node (P3label) at (9,3) {$P_3$};
                \node (P4label) at (11,0.5) {$P_4$};
                
            \end{scope}
        \end{tikzpicture}
        \caption{An example showing the decomposition of $C$ into directed paths $P_1,\ldots,P_4$. In this example, $P_2$ has positive weight.}
        \label{fig:decomp}
    \end{figure}
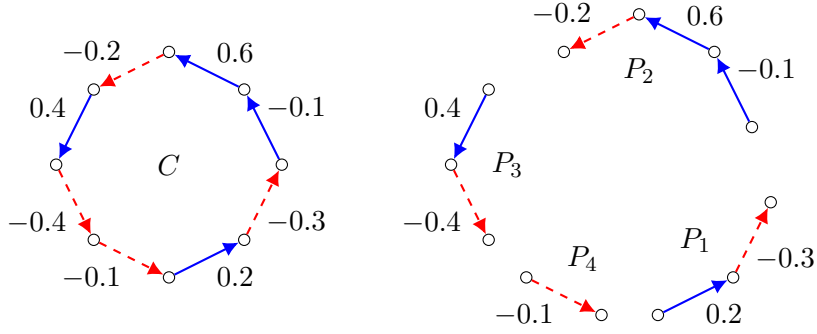
    
    Now, since $C$ has positive total weight, there is a directed path $P \in \{P_1,\ldots,P_d\}$ of positive total weight. 
    Without loss of generality, let $P = \{1,2,\dots,k+1\}$. Thus in the envy graph $G_{\mbs{\mu}^t}$ we have
        \begin{equation}\label{eq:positive-path}
        w_{\mbs{\mu}^t}(P) \ =\  \sum_{i=1}^{k} w_{\mbs{\mu}^t}(i,i+1) \ >\  0.
        \end{equation}
   
    Construct a matching $\mathcal{M}^t = \{(i,\omega_i^t)\}_{i\in I}$ in the following manner. 
    For each agent $i\ge k+1$, set $\omega^t_i=\mu^t_i$; that is, the end-vertex of the path $P$ and all agents not on $P$
    are matched to the same item in $\mathcal{M}^t$ as in $M^t$.
    For each agent $i\le k-1$, let $\omega^t_i=\mu^t_{i+1}$, that is in the allocation $\mathcal{M}^t$ agent $i$ receives the 
    item that agent $i+1$ receives in $M^t$. Finally, for agent $k$ let $\omega^t_k=M^{t+1}_k$; 
    that is, in $\mathcal{M}^t$ agent $k$ receives the 
    item it would have received in the next round in $M^{t+1}$.
  
  Observe that every item allocated by $\mathcal{M}^t$ was available for allocation in round $t$ and, thus,
  it was a feasible allocation to select in round $t$.
  Next let's compare the relative values of $\mathcal{M}^t$ and $M^t$ under the original valuations $\mbs{v}$.
  To do this, observe that by definition of $\mathcal{M}^t$ we have
    \begin{align}\label{eq:mathcal-m}
  v(\mathcal{M}^t)-v(M^t) &=  \sum_{i=1}^k \left( v_i(\omega^t_i)-v_i(\mu^t_i)\right) \nonumber \\
  &=  \sum_{i=1}^{k-1} \left( v_i(\omega^t_i)-v_i(\mu^t_i)\right) + \left( v_k(\omega^t_k)-v_k(\mu^t_k)\right) \nonumber \\
    &=  \sum_{i=1}^{k-1} \left( v_i(\mu^t_{i+1})-v_i(\mu^t_i)\right) + \left( v_k(\mu^{t+1}_k)-v_k(\mu^t_k)\right).
               \end{align}
But $(k,k+1)$ is a red arc in $G_{\mbs{\mu}^t}$. Therefore, it must be the case that
$v_{k}(\mu^{t+1}_{k}) > v_{k}(\mu^t_{k+1})$. Plugging this into (\ref{eq:mathcal-m}) gives
    \begin{align}\label{eq:mathcal-m2}
  v(\mathcal{M}^t)-v(M^t)      &\ >\   \sum_{i=1}^{k-1} \left( v_i(\mu^t_{i+1})-v_i(\mu^t_i)\right) + \left( v_k(\mu^{t}_{k+1})-v_k(\mu^t_k)\right)  \nonumber \\
         &\ =\   \sum_{i=1}^{k} \left( v_i(\mu^t_{i+1})-v_i(\mu^t_i)\right).
           \end{align}
           But, by definition, $w_{\mbs{\mu}^t}(i,i+1)=v_i(\mu^t_{i+1})-v_i(\mu^t_i)$. So, together (\ref{eq:positive-path}) and (\ref{eq:mathcal-m2}) imply 
      \begin{equation}
  v(\mathcal{M}^t)-v(M^t) 
     \ >\  \sum_{i=1}^{k} w_{\mbs{\mu}^t}(i,i+1) \ > \ 0.
  \end{equation}
 
Thus $\mathcal{M}^t$ has greater weight than $M^t$ under the original valuations $\mbs{v}$.
  This contradicts the optimality of $M^t$.  \qedhere
\end{itemize}
\end{proof}

Claim~\ref{clm:allocation} shows that the allocation $\mathcal{A}^{\mbs{v}}$ produced by the algorithm on the original 
instance is an envy-freeable allocation in the modified instance. 
We next show that for this modified 
valuation profile the subsidy required is at most $1$ for each agent. In particular the
total subsidy is at most $n-1$.
\begin{clm}\label{clm:subsidybound}
For the envy-freeable allocation $\mathcal{A}^{\mbs{v}}$ the subsidy to each agent is at most $1$
for the modified valuation profile $\mbs{\bar{v}}$.
\end{clm}
\begin{proof}
Take the valuation profile $\mbs{\bar{v}}$ and the allocation $\mathcal{A}^{\mbs{v}} = \{A^{\mbs{v}}_1,\ldots,A^{\mbs{v}}_n\}$.
We claim that for any arc $(i, k)$ its modified weight $\bar{w}_{\mbs{A^{v}}}(i,k)$ in the the envy graph is \textit{at least} $-1$. 
To prove this take any pair of agents $i$ and $k$. Then
\begin{align}\label{eq:modified-weights}
\bar{w}_{\mbs{A^{v}}}(i,k)  
&\ =\  \bar{v}_i(A^{\mbs{v}}_{k})-\bar{v}_i(A^{\mbs{v}}_i) \nonumber \\
&\ =\  \sum_{t=1}^T  \bar{v}_i(\mu^t_{k}) - \sum_{t=1}^T \bar{v}_i(\mu^t_i) \nonumber\\
&\ =\  \sum_{t=1}^T \bar{v}_i(\mu^t_{k})- \sum_{t=1}^T v_i(\mu^t_i)  \nonumber\\
&\ =\  \sum_{t=1}^{T-1}  \max (v_i(\mu^t_{k}),v_i(\mu^{t+1}_{i})) + v_i(\mu^T_i)-\sum_{t=1}^T v_i(\mu^t_i) \nonumber\\
&\ \ge\ \sum_{t=1}^{T-1}  v_i(\mu^{t+1}_{i})- \sum_{t=1}^{T-1} v_i(\mu^t_i).
\end{align}
We can simplify (\ref{eq:modified-weights}) and lower bound it via a telescoping sum:
\begin{align}\label{eq:minus-one}
\bar{w}_{\mbs{A^{\mbs{v}}}}(i,k)  
&\ \ge\  \sum_{t=1}^{T-1} \left( v_i(\mu^{t+1}_i) -  v_i(\mu^{t}_{i}) \right) \nonumber \\
&\ =\    v_i(\mu^{T}_{i}) - v_i(\mu^1_i)   \nonumber  \\
&\ \ge\  - v_i(\mu^1_i) \nonumber  \\
&\ \ge\  -1.
 \end{align}
Now by Claim~\ref{clm:allocation}, the allocation $\mathcal{A}^{\mbs{v}}$ is envy-freeable with respect to the valuations $\mbs{\bar{v}}$.
Applying Lemma~\ref{thm:one-dollar-less}, because the arc weights are lower bounded by $-1$ the subsidy
required per agent is then at most one for the modified valuation profile $\mbs{\bar{v}}$. 

Finally, since there is an agent whose payment is 0, the total subsidy required is upper bounded by $n-1$.
\end{proof}

The following claim shows that, for any agent, the subsidy for the original valuation profile is at most 
the subsidy required for the modified valuation function.

\begin{clm}\label{clm:subsidy}
For the allocation $\mathcal{A}^{\mbs{v}}$ the subsidy required by an agent
given valuation profile $\mbs{v}$ is at most the subsidy required 
given valuation profile $\bar{\mbs{v}}$. 
\end{clm}
\begin{proof}
By Observation~\ref{obs:equal}, $v_i(j) = \bar{v}_i(j)$ for any $j\in A^{\mbs{v}}_i$.
Therefore, by additivity, 
\begin{equation}\label{eqn:vbar-equal}
\bar{v}_i(A^{\mbs{v}}_i) \ =\  \sum_{j\in A^{\mbs{v}}_i} \bar{v}_i(j) \ =\  \sum_{j\in A^{\mbs{v}}_i} v_i(j) \ =\ v_i(A^{\mbs{v}}_i).
\end{equation}
On the other hand, Observation~\ref{obs:upper} states that $v_i(j) \le \bar{v}_i(j)$ for any $j\notin A^{\mbs{v}}_i$.
Thus, for any pair $i$ and $k$ of agents, we have 
\begin{equation}\label{eqn:vbar-upper}
    \bar{v}_i(A^{\mbs{v}}_{k})  \ =\ \sum_{j\in A^{\mbs{v}}_{k}}\bar{v}_i(j) 
     \ \geq \ \sum_{j\in A^{\mbs{v}}_{k}}v_i(j) 
    \ =\  v_i(A^{\mbs{v}}_{k}).
\end{equation}
Combining (\ref{eqn:vbar-equal}) and (\ref{eqn:vbar-upper}) gives 
\begin{align*}
\bar{w}_{\mbs{A^{v}}}(i,k)  \ =\ \bar{v}_i(A^{\mbs{v}}_k)-\bar{v}_i(A^{\mbs{v}}_{i}) 
\ \ge\ v_i(A^{\mbs{v}}_{k})-v_i(A^{\mbs{v}}_{i}) \ =\ w_{\mathcal{A}^{\mbs{v}}}(i,k). \end{align*}
Consequently, the weight of any arc $(i,k)$ in the envy graph with the modified valuation profile is at least
its weight with the original valuation profile. Therefore the weight of any path in the envy graph is
higher with the modified valuation profile than with the original valuation profile.
The claim follows.
\end{proof}

Together Claims~\ref{clm:subsidybound} and~\ref{clm:subsidy} give our main result.

\newtheorem*{thm:main}{Theorem \ref{thm:main}}
\begin{thm:main}
For additive valuations there is 
an envy-freeable allocation where the subsidy to each agent is at most one dollar.
(This allocation is also EF1, balanced, and can be computed in polynomial time.)
\end{thm:main}

\section{Bounding the Subsidy for Monotone Valuations} \label{s:monotone}

We now consider the much more general setting where the valuations of the agents 
are arbitrary monotone functions. That is, the only assumptions we impose are that $v_i(S) \leq v_i(T)$ 
when $S\subseteq T$ and the basic assumption that $v_i(\emptyset) = 0$. 
Without loss of generality, we may scale the valuations so that 
the marginal value of each item for any agent never exceeds one dollar. Our goal in this section is to show that 
there is an envy-freeable allocation in which the total subsidy required for envy-freeness is at most $2(n-1)^2$.
In particular, the total subsidy required is independent of of the number of items~$m$. 
When $m > 2(n-1)$ this bound beats the bound $(n-1)\cdot m$ of \cite{HS19} for additive valuations
 described in Theorem~\ref{thm:hsthmub}\ and, more importantly, it applies to the far more general 
class of arbitrary monotone valuations.

Our method to compute the desired envy-freeable allocation begins with finding an EF1 allocation. 
The well-known \textit{envy-cycles} algorithm of \citet{LMM04} finds such an allocation in polynomial time 
given oracle access to the valuations, under the same mild conditions on the valuations. For completeness, we briefly 
describe the envy-cycles algorithm. The algorithm proceeds in a sequence of $m$ rounds, allocating one item in 
each round. At any point during the algorithm, we denote by $G$ the envy graph corresponding to the current 
allocation, and by $H$ the subgraph of $G$ that consists of all the agents and only the arcs that have positive weight, that is,
positive envy. We call $H$ the {\em auxiliary graph} of $G$. The algorithm relies on the following lemma.
\begin{lem}\label{lem:envy-cycles}
\cite{LMM04} For any partial allocation $\mathcal{A}$ with auxiliary graph $H$, there is another partial 
allocation $\mathcal{A}'$ with auxiliary graph $H'$ such that
\vspace{-.25cm}
\begin{itemize}[noitemsep]
    \item $H'$ is acyclic.
    \item For each agent $i$, the maximum weight of an outgoing arc from $i$ is less in $\mathcal{A}'$ than in $\mathcal{A}$.
\end{itemize}
\end{lem}
The basic idea of the algorithm then is to maintain the following two invariants: (i) at each step, 
the partial allocation is EF1, and (ii) at the start and end of each round, the 
auxiliary graph $H$ is acyclic. Since the auxiliary graph is a directed acyclic graph at the start of each round, it has a 
source vertex. The algorithm simply chooses this vertex and allocates the next item to the corresponding agent. 
Because no other agent envies this agent before this item is allocated, the envies of the other agents are bounded 
by the value of this item (so the allocation of this item maintains the EF1 invariant). Next, the algorithm identifies 
a directed cycle (if one exists) in the auxiliary graph $H$ and redistributes bundles by rotating them around this cycle. 
It is easy to see that the EF1 guarantee is maintained after this redistribution of the bundles, and that the number of 
arcs in $H$ strictly decreases. All cycles in $H$ are then eliminated in sequence until $H$ is acyclic 
and the round ends. When all items have been allocated, the final allocation is EF1.

This immediately raises the question of whether the resulting allocation is envy-freeable. By 
Claim~\ref{clm:envy-freeable-EF1}, we know that if an allocation is both envy-freeable and EF1, then 
the total subsidy required for envy-freeness is $(n-1)^2$, since the weight of any path is at most $n-1$. 
Unfortunately, it is possible that the allocation output by the envy-cycles algorithm is not envy-freeable. 
However, we show that an EF1 allocation can still be used to produce an envy-freeable allocation that requires 
only a small increase in the subsidy! Specifically, the following key lemma shows that if we begin by fixing the 
bundles of an EF1 allocation and then redistribute these bundles to produce an envy-freeable allocation, the weight 
of any path increases to at most $2(n-1)$. By Theorem \ref{thm:hsthm1}, an envy-freeable allocation can be found 
by computing a maximum-weight matching.

\begin{lem} \label{lem:redistribute}
Let $\mathcal{A}$ be an EF1 allocation, and $\mathcal{B}$ be the envy-freeable allocation corresponding to a 
maximum-weight matching between the agents and the bundles of $\mathcal{A}$. Then $\mathcal{B}$ can be 
made envy-free with a subsidy of at most $2(n-1)$ to each agent.
\end{lem}
\begin{proof}
Let $\mathcal{A} = \{A_1,\ldots,A_n\}$ be an EF1 allocation. So, for any pair $i$ and $k$ of 
agents, $v_i(A_k) - v_i(A_i) \leq 1$. Let $\pi$ be a permutation of the bundles that 
maximizes $\sum_i v_i(A_{\pi(i)})$. Then, by Theorem~\ref{thm:hsthm1}, the 
allocation $\mathcal{B} = \{B_{1},\ldots,B_{n}\} = \{A_{\pi(1)},\ldots,A_{\pi(n)}\}$ is envy-freeable. 
Next, let $P$ be a directed path in the envy graph $G_{\mathcal{B}}$. Without loss of generality, $P=\{1,2,\dots,r\}$ for 
some $r\ge 2$. Our goal is to show that the weight of $P$ in $G_{\mathcal{B}}$ is at most $2(n-1)$. Clearly, the 
weight of $P$ in $G_{\mathcal{A}}$ is at most $n-1$. Consider an arc $(i,i+1)$ of $P$. Since $\mathcal{A}$ is EF1, for 
any agent $k$, we have $v_i(A_k) - v_i(A_i) \leq 1$. Now, agent $i+1$ receives the bundle of 
agent $\pi(i+1)$ in the redistributed allocation $\mathcal{B}$. We have $v_i(A_{\pi(i+1)}) - v_i(A_i) \leq 1$
and, thus, $v_i(B_{i+1}) - v_i(A_i) \leq 1$. It follows that:
\begin{align}\label{eqn:x}
    w_{\mathcal{B}}(P) &= \sum_{(i,k)\in P} w_{\mathcal{B}}(i,k) & \nonumber \\
    &= \sum_{i=1}^{r-1} \left(v_i(B_{i+1}) - v_i(B_i)\right) & \nonumber \\
    &= \sum_{i=1}^{r-1} \left(v_i(B_{i+1}) - v_i(A_i) + v_i(A_i) - v_i(B_i)\right) &\nonumber \\
    &\le \sum_{i=1}^{r-1} \left(1 + v_i(A_i) - v_i(B_i)\right) &\nonumber \\
    &\leq (n-1) + \sum_{i=1}^{r-1} \left(v_i(A_i) - v_i(B_i)\right)
\end{align}
To complete the proof, it remains to show that $\sum_{i=1}^{r-1} \left(v_i(A_i) - v_i(B_i)\right)$ is at most $n-1$. Together with (\ref{eqn:x}), this implies that $w_{\mathcal{B}}(P) \leq 2(n-1)$.

Since $\pi$ maximizes $\sum_i v_i(A_{\pi(i)})$, we have $\sum_i v_i(B_i) \geq \sum_i v_i(A_i)$. The key observation is 
that, while the sum of values of the bundles received by all agents increases when we redistribute the bundles 
from $\mathcal{A}$ to $\mathcal{B}$, the value of the bundle received by any single agent 
increases by {\em at most} one because $\mathcal{A}$ is EF1. This then constrains the 
amount by which the total value for any subset of agents 
can \textit{decrease}. Specifically, let $R \subseteq I$ be the set of agents $i$ that 
receive a bundle $B_i$ of smaller value than $A_i$, that is, $R = \{i\in I:v_i(B_i) < v_i(A_i)\}$. 
Let $S = I\setminus R$, so $S = \{i\in I:v_i(B_i) \geq v_i(A_i)\}$.

Now, we have two cases to consider.
\begin{itemize}
    \item[(i)] \textit{$|R| = 0$.}
    
    Then $\sum_{i=1}^{r-1} \left(v_i(A_i) - v_i(B_i)\right) \leq 0$ and the result follows.
    
    \item[(ii)] \textit{$|R|\geq 1$.}
    
    Then $|S| \leq n-1$, and we have
    \begin{align*}
        \sum_{i\in[r-1]} \left(v_i(A_i) - v_i(B_i)\right) &= \sum_{i\in[r-1]\cap R} \left(v_i(A_i) - v_i(B_i)\right) + \sum_{i\in[r-1]\cap S} \left(v_i(A_i) - v_i(B_i)\right) &\\
        &\leq \sum_{i\in[r-1]\cap R} \left(v_i(A_i) - v_i(B_i)\right) &\\
        &\leq \sum_{i\in R} \left(v_i(A_i) - v_i(B_i)\right) &\\
        &\leq \sum_{i\in S} \left(v_i(B_i) - v_i(A_i)\right) &\\
        &\leq n-1.
    \end{align*}
    The second to last inequality says that the total decrease in value for agents in $R$ is at most the total increase in value for agents in $S$ (since $B$ is an optimal redistribution of the bundles). The final inequality follows from the fact that $|S| \leq n-1$ and for each $i\in S$, $v_i(B_i) - v_i(A_i) \leq 1$ since $\mathcal{A}$ is EF1. \qedhere
\end{itemize}
\end{proof}

Together, Lemmas~\ref{lem:envy-cycles} and~\ref{lem:redistribute} bound the total subsidy 
sufficient for envy-freeness when the valuation functions are monotone.
\newtheorem*{thm:monotone}{Theorem \ref{thm:monotone}}
\begin{thm:monotone}
For monotonic valuations there is 
an envy-freeable allocation where the subsidy to each agent is at most $2(n-1)$ dollars. (Given a valuation oracle, 
this allocation can be computed in polynomial time.) \qed
\end{thm:monotone}

\bibliography{res}
\end{document}